\documentclass[oneside,11pt]{article}

\usepackage{graphicx}
\usepackage{amsmath,amssymb} 
\usepackage{algorithm}
\usepackage{algorithmic}

\newcommand{\Id}{\mathtt{I}}

\newcommand{\rank}[1]{\ensuremath{\mathrm{rank}\left(#1\right)}}
\newcommand{\norm}[1]{\ensuremath{\left\| #1 \right\|}}

\newtheorem{problem}{Problem}
\newtheorem{theorem}{Theorem}
\newtheorem{lemma}{Lemma}

\newenvironment{proof}[1][Proof]{\begin{trivlist}
\item[\hskip \labelsep {\bfseries #1}]}{\end{trivlist}}

\newcommand{\qed}{\nobreak \ifvmode \relax \else
      \ifdim\lastskip<1.5em \hskip-\lastskip
      \hskip1.5em plus0em minus0.5em \fi \nobreak
      \vrule height0.5em width0.5em depth0.1em\fi}

\begin{document}

\title{\LARGE\bf Advancing Matrix Completion by Modeling Extra Structures beyond Low-Rankness}
\author{ \bf{Guangcan Liu} \\
         Department of Statistics and Biostatistics\\
         Department of Computer Science\\
         Rutgers University\\
         Piscataway, NJ 08854, USA\\
       \texttt{guangcan.liu@rutgers.edu}\\\\
       \and
         \bf{Ping Li} \\
         Department of Statistics and Biostatistics\\
         Department of Computer Science\\
       Rutgers University\\
        Piscataway, NJ 08854, USA\\
       \texttt{pingli@stat.rutgers.edu}}

\date{}

\maketitle
\begin{center} {\bf Abstract } \end{center}

A well-known method for completing low-rank matrices based on convex optimization has been established by Cand{\`e}s and Recht~\cite{Candes:2009:math}. Although theoretically complete, the method may not entirely solve the low-rank matrix completion problem. This is because the method captures only the low-rankness property which gives merely a constraint that the data points locate on some low-dimensional subspace, but generally ignores the extra structures which specify in more detail how the data points locate on the subspace. Whenever the geometric distribution of the data points is not uniform, the coherence parameters of data might be large and, accordingly, the method might fail even if the latent matrix to recover is fairly low-rank. To better handle non-uniform data, in this paper we propose a model termed Low-Rank Factor Decomposition (LRFD), which imposes an additional restriction that the data points must be represented as linear combinations of the bases in a given dictionary. We show that LRFD can well handle non-uniform data, provided that the dictionary is configured properly: We mathematically prove that if the dictionary itself is low-rank then LRFD is immune to the coherence parameters which might be large on non-uniform data. This provides an elementary principle for learning the dictionary in LRFD and, naturally, leads to a practical algorithm for advancing matrix completion. Extensive experiments on randomly generated matrices and motion datasets show encouraging results.

\newpage\clearpage

\section{Introduction} In modern applications such as \emph{structure from motion}, very often one needs to restore the missing entries of a matrix, i.e., \emph{matrix completion}~\cite{johnson:1990:mc}. In general, given no presumptions about the nature of the entries, matrix completion is virtually impossible as the missing entries can be of arbitrary values. Due to the low-rankness nature of today's high-dimensional data, a commonly adopted assumption is that the latent matrix we want to recover is fairly low-rank, resulting in the so-called \emph{low-rank matrix completion} problem, which is formulated as follows:
\begin{problem}[Low-Rank Matrix Completion]\label{pb:lmc}
Suppose we have a data matrix $X\in\mathbb{R}^{m\times{}n}$, which is known only on a fraction of its entries:
\begin{eqnarray*}
[X]_{ij}=[L_0]_{ij}, \forall{}(i,j)\in\Omega,
\end{eqnarray*}
where $L_0\in\mathbb{R}^{m\times{}n}$ is a low-rank matrix each column of which is a data point lying on some low-dimensional subspace, $[\cdot]_{ij}$ denotes the $(i,j)$th entry of a matrix, and $\Omega\subset{}\{1,\cdots,m\}\times\{1,\cdots,n\}$ is an index set consisting of the locations of the observed entries. Given the incomplete matrix $X$ (and the index set $\Omega$), can we exactly recover the latent matrix $L_0$ in a scalable way?
\end{problem}

There is a large community that explores the above problem using various statistical tools, e.g.,~\cite{Candes:2009:math,CandesPIEEE,Mohan:2010:isit,Recht2008,Negahban:2012:JMLR,rahul:jlmr:2010,Srebro05generalizationerror,wang:arxiv:2014,liuj:tpami:2013,chen:icml:2011}. Of all those notable contributions, the most fundamental and significant one is probably the convex optimization based method established by Cand{\`e}s and Recht~\cite{Candes:2009:math}. For the ease of presentation, we shall call this method as ``CONO" (CONvex Optimization) for short. CONO tells us for sure that, when the low-rank matrix $L_0$ is meanwhile \emph{incoherent} (i.e., with low coherence parameters), $L_0$ can be exactly recovered by using the following convex, parameter-free, and potentially scalable program:
\begin{eqnarray}\label{eq:cono}
\min_{L}\|L\|_*, &\textrm{s.t.}&\mathcal{P}_{\Omega}(X-L)=0,
\end{eqnarray}
where $\|\cdot\|_*$ is the \emph{nuclear norm}~\cite{phd_2002_nuclear,siam_2010_minirank} of a matrix, i.e., the sum of the singular values of a matrix, and $\mathcal{P}_{\Omega}$ denotes the orthogonal projection onto the linear space of matrices supported on $\Omega$. Besides of its completeness in theory, CONO also has good empirical performance and is therefore widely regarded as a milestone in the history of matrix completion.

Nevertheless, CONO cannot be the best solution to the low-rank matrix completion Problem~\ref{pb:lmc}. Indeed, the method might be unsuccessful even when the latent matrix $L_0$ is strictly low-rank and the locations of missing entries are selected uniformly at random. This is because CONO captures only the low-rankness property of $L_0$, but essentially ignores the \emph{extra structures} which are critical to the success of recovery: Given the low-rankness constraint that the data points (i.e., columns vectors of $L_0$) locate on a low-dimensional subspace, it is quite normal that the data may have some extra structures which specify in more detail \emph{how} the data points locate on the subspace, as illustrated in Figure~\ref{fig:extra}. Notice that the extra structures are essentially nonlinear and hard to parameterize. Therefore, we shall not adopt parametric models to describe and explore each extra structure in a particular way, but instead generally divide all cases shown in Figure~\ref{fig:extra} into two categories:
\begin{itemize}
\item[1)] \textbf{Uniform data}: The data points \emph{uniformly} locate on a low-dimensional subspace, as shown in Figure~\ref{fig:extra}(a).
\item[2)] \textbf{Non-uniform data}: The data points \emph{non-uniformly} locate on a low-dimensional subspace, as shown in Figure~\ref{fig:extra}(b) $\sim$ Figure~\ref{fig:extra}(e).
\end{itemize}
\begin{figure}[h!]
\begin{center}
\includegraphics[width=0.95\textwidth]{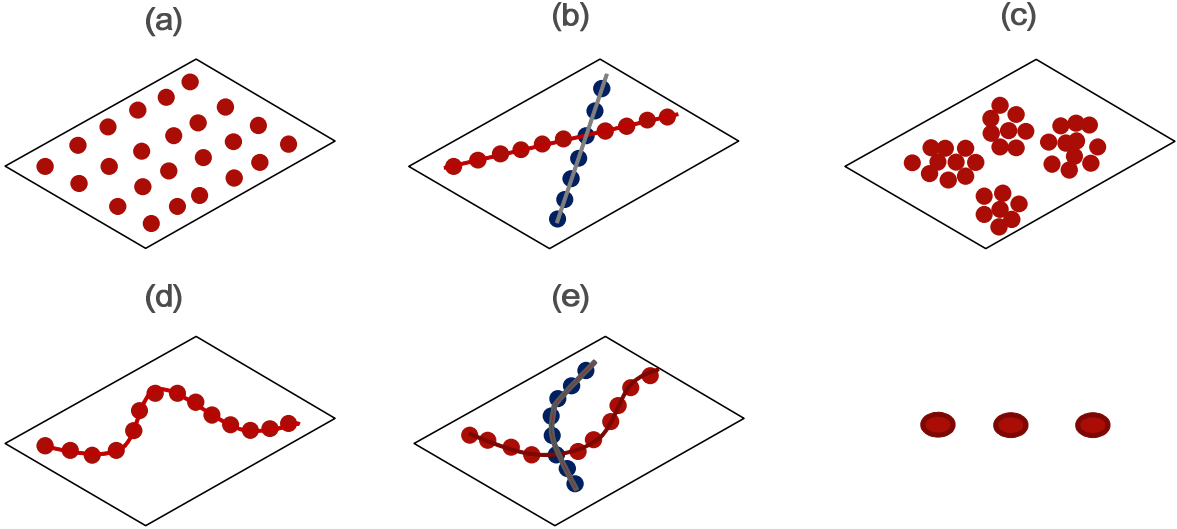}
\caption{Illustrating the extra structures beyond low-rankness. Each column of the data matrix $L_0$ is a data point. Given the constraint that $L_0$ is low-rank, i.e., the data points locate on a low-dimensional subspace, more specific situation could be: (a) The data points uniformly distribute on the subspace, (b) the data points have a mixture structure of multiple ``small'' subspaces inside the ``large'' subspace, (c) the data points form multiple ``ball-like'' clusters, (d) the data points lie on a nonlinear manifold inside the subspace, (e) the data points follow a mixture structure of multiple nonlinear manifolds inside the subspace, etc.}\label{fig:extra}
\end{center}
\end{figure}

For the uniform case as in Figure~\ref{fig:extra}(a), CONO is probably the best method for low-rank matrix completion. Nevertheless, uniform data actually seldom exist in reality and CONO might not work well on non-uniform data. The reason is that the coherence parameters of non-uniform data might be large, and thus CONO might fail to recover $L_0$ even when $L_0$ is fairly low-rank. Even more, non-uniform data are ubiquitous in realistic areas such as computer vision. For example, it is known that the data matrix of trajectories of motion objects provably follows a mixture structure of multiple subspaces as in Figure~\ref{fig:extra}(b)~\cite{ijcv_1998_factor}. Anyway, uniform data is after all a special case of non-uniform data, and thus it is undoubtedly significant to study the matrix completion problem in the context of non-uniform data.

To accomplish an advanced solution to the low-rank matrix completion Problem~\ref{pb:lmc} in the context of non-uniform data, in this paper we propose to consider a generalized version of \eqref{eq:cono}, called as \emph{Low-Rank Factor Decomposition} (LRFD) for the convenience of citation:
\begin{eqnarray}\label{eq:mc}
\min_{Z}\|Z\|_*, &\textrm{s.t.}&\mathcal{P}_{\Omega}(X-AZ)=0,
\end{eqnarray}
where $A\in\mathbb{R}^{m\times{}d}$ is a dictionary matrix constructed or learnt in advance (the choice of the dictionary size $d$ is immaterial). Note here that, unlike in CONO, in our LRFD it is $AZ^*$ that reconstructs $L_0$ (assume $Z^*$ is the minimizer to \eqref{eq:mc}). It is easy to see that \eqref{eq:mc} falls back to \eqref{eq:cono} when $A = \Id$ (identity matrix). So it could be regarded that LRFD is a generalization of CONO.

To well handle non-uniform data, the dictionary matrix $A$ should be chosen properly. We shall mathematically prove that if the dictionary itself is low-rank then LRFD is immune to the coherence parameters which might be large on non-uniform data. This provides an elementary principle for learning the dictionary in LRFD. Subsequently, we devise a practical algorithm to obtain proper dictionaries in unsupervised environments. Our extensive experiments on randomly generated matrices and motion datasets show encouraging results. In summary, our contributions include:
\begin{itemize}
\item[$\bullet$] We propose to improve low-rank matrix completion by modeling the extra structures possibly existing in data. To our knowledge, we are the first to pursue this direction in the community of matrix completion. Furthermore, we establish a generic model termed LRFD, some elementary theories and a practical algorithm for resolving the problem of restoring a low-rank (yet non-uniform) matrix from its incomplete versions.

\item[$\bullet$] The idea of replacing a variable $L$ with the product of two variables, saying $AZ$, is essentially the spirit of \emph{matrix factorization} which has been discussed for long, e.g.,~\cite{nips:WeimerKLS07,tpami_2013_lrr,Srebro05generalizationerror,Proc:Wang_Li_SDM10,chen:icml:2012}. In that sense, the investigations of this paper help to understand why the factorization techniques could be effectual.
\item[$\bullet$] While the concept of \emph{coherence} is now standard and widely used in various literatures, e.g.,~\cite{Candes:2009:JournalACM,xu:2010:nips}, there is a lack of studies about the \emph{physical} regime that affects the behaviors of coherence parameters. This paper shows that the coherence parameters are related in nature to the geometric distribution of data points: The more non-uniformly the data points distribute, the larger the coherence parameters could be.
\end{itemize}
\section{Summary of Main Notations}\label{sec:notation}
Capital letters such as $M$ are used to represent matrices, and accordingly, $[M]_{ij}$ denotes its $(i,j)$th entry. The particular symbol $(\cdot)^+$ denotes the Moore-Penrose pseudo-inverse of a matrix, i.e., $M^+=V_M\Sigma_M^{-1}U_M^T$ for any matrix $M$ with SVD\footnote{In this paper, SVD always refers to skinny SVD. For a rank-$r$ matrix $M\in\mathbb{R}^{p\times{}q}$, its SVD is of the form $U_M\Sigma_MV_M^T$, where $U_M\in\mathbb{R}^{p\times{}r},\Sigma_M\in\mathbb{R}^{r\times{}r}$ and $V_M\in\mathbb{R}^{q\times{}r}$.} $U_M\Sigma_MV_M^T$. Letters $U$, $V$, $\Omega$ and their variants (complements, subscripts, etc.) are reserved for column space, row space and index set, respectively. We shall abuse the notation $U$ to denote the linear space spanned by the columns of $U$. The projection onto the column space, $U$, is denoted by $\mathcal{P}_U$ and given by $\mathcal{P}_U(M)=UU^TM$. We shall also abuse the notation $\Omega$ to denote the linear space of matrices supported on $\Omega$, and use $\mathcal{P}_{\Omega}$ and $\mathcal{P}_{\Omega^{\bot}}$ to respectively denote the projections onto $\Omega$ and $\Omega^c$ (i.e., the complement of $\Omega$) such that $\mathcal{P}_{\Omega}+\mathcal{P}_{\Omega^{\bot}}=\mathcal{I}$, where $\mathcal{I}$ is the identity operator.

Three types of matrix norms are used in this paper, and they are all functions of the singular values: 1) the operator norm or 2-norm (i.e., the largest singular value) denoted by $\|M\|$, 2) the Frobenius norm (i.e., the square root of the sum of squared singular values) denoted by $\|M\|_F$ and 3) the nuclear norm or trace norm (i.e., the sum of singular values) denoted by $\|M\|_*$. The only used vector norm is the $\ell_2$ norm, which is denoted by $\|\cdot\|_2$.

The letter $\mu$ and its variants are reserved to denote the coherence parameters of a matrix. We also reserve two lowercase letters, $m$ and $n$, to respectively denote the data dimension and the number of data points, and we use the following two symbols throughout this paper:
\begin{eqnarray*}
n_1 = \max(m,n) &\textrm{and}&n_2 = \min(m,n).
\end{eqnarray*}

\section{Analysis, Theory and Algorithm}
In this section, we shall try to answer the following two questions: (1) Why CONO might not work well on non-uniform data ? (2) How to choose the dictionary matrix $A$ in LRFD?
\subsection{Why CONO Might Fail on Non-Uniform Data?}
To get a definite answer to the question highlighted above, we introduce below the concept of \emph{coherence} and investigate the physical regime that affects the behaviors of coherence parameters.

The definition of coherence adopted by this paper is the same as~\cite{Candes:2009:math,Candes:2009:JournalACM}. For a matrix $M\in\mathbb{R}^{p\times{}q}$ with rank $r$ and SVD $U_M\Sigma_MV_M^T$, there are two coherence parameters, $\mu_1$ and $\mu_2$, which are useful to characterize the statistical properties of the matrix. The first coherence parameter, $1\leq\mu_1\leq{}p$, which captures the statistical properties of the \emph{column space} identified by $U_M$, is defined as
\begin{eqnarray}\label{eq:u1}
\mu_1(M) = \frac{p}{r}\max_i\|U_M^Te_i\|_2^2,
\end{eqnarray}
where $e_i$ denotes the $i$th standard basis. The second coherence parameter, $1\leq\mu_2\leq{}q$, which characterizes the \emph{row space} identified by $V_M$, is defined as
\begin{eqnarray}\label{eq:u2}
\mu_2(M) = \frac{q}{r}\max_i\|V_M^Te_i\|_2^2.
\end{eqnarray}

\begin{figure}[h!]
\begin{center}
\includegraphics[width=0.95\textwidth]{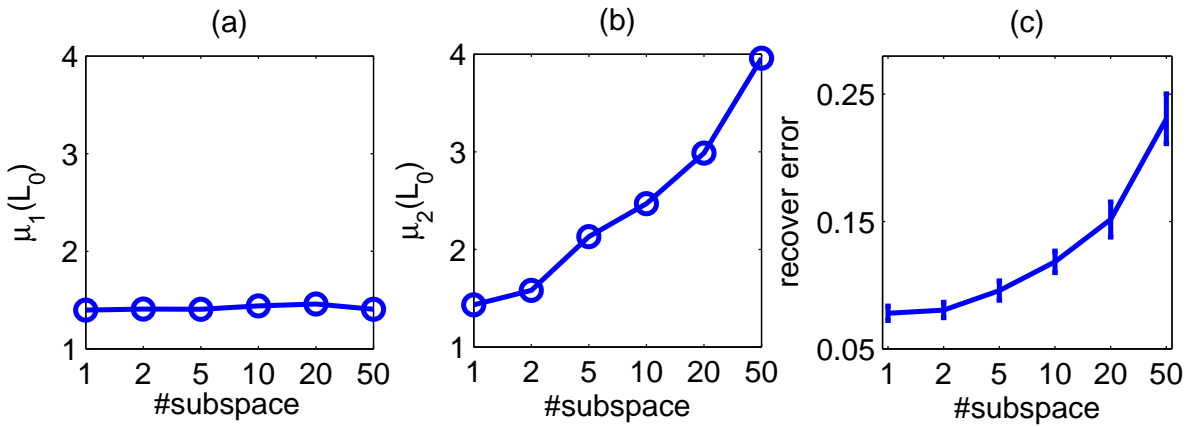}
\caption{Exploring the properties of coherence parameters, using randomly generated matrices. The size of $L_0$ is fixed to be $500\times500$. The underlying subspace number $k$ varies from 1 to 50. We set the dimension of each subspace as $100/k$, and thus $L_0$ has a fixed rank of 100. (a) The first coherence parameter $\mu_1(L_0)$ vs subspace number. (b) The second coherence parameter $\mu_2(L_0)$ vs subspace number. (c) The performance of CONO vs subspace number. For the matrix completion experiments in (c), the percentage of missing entries is fixed to be 45\%. The recover error is computed as $\|\hat{L}_0-L_0\|_F/\|L_0\|_F$, where $\hat{L}_0$ is an estimate of $L_0$. The numbers shown in (c) are collected from 100 random trials.}\label{fig:icu2}
\end{center}
\end{figure}

Since the behaviors of data points could affect the row space $V_M$, the second coherence parameter $\mu_2$ may somehow depend on the geometric distributions of the data points. To confirm, we consider for exploration the mixture structure shown in Figure~\ref{fig:extra}(b), which is about the phenomenon that the data points in $L_0$ are sampled from $k$ number of subspaces, i.e., $L_0=[L_0^{(1)},\cdots,L_0^{(k)}]$, where $L_0^{(i)}$ is the matrix of data points from the $i$th subspace. While the rank of $L_0$ is fixed and the underlying subspace number $k$ goes large, Figure~\ref{fig:icu2}(b) shows that the second coherence parameter $\mu_2(L_0)$ keeps increasing. To see why the second coherence parameter increases with the cluster number underlying $L_0$, please refer to~\cite{liu:arxiv:2014}.

Among other things, the information revealed by Figure~\ref{fig:icu2}(a) is remarkable and useful: The first coherence parameter $\mu_1(L_0)$ is immune to the variation of the underlying subspace number. This is actually natural, because the behaviors of the data points can only affect the row space, while $\mu_1$ is defined on the column space\footnote{Notice that $\mu_1$ could be also large if the row vectors of $L_0$ own some structures beyond low-rankness. Such kind of data exist widely in text domain and we leave this as future work.}. Analogously, we have the following doctrines that depict the coherence parameters in general:
\begin{itemize}
\item[$\bullet$] The first coherence parameter $\mu_1(L_0)$ is always small, in despite of whether or not the geometric distribution of the data points is uniform.
\item[$\bullet$] The second coherence parameter $\mu_2(L_0)$ is small on the uniform data, but could be large on the non-uniform cases such as Figure~\ref{fig:extra}(b).
\end{itemize}

Now the answer to the question highlighted in the beginning of this subsection is clear. Namely, the analysis in~\cite{Candes:2009:math} illustrates that CONO prefers the cases where both $\mu_1$ and $\mu_2$ are small. Nevertheless, such an expectation might not be true as the second coherence parameter $\mu_2$ could be large on non-uniform data and, accordingly, the recovery performance of CONO might be unsatisfactory even when $L_0$ is strictly low-rank. To verify this assertion, we have executed lots of numerical experiments. As we can see from Figure~\ref{fig:icu2}(c), CONO degrades with the enlargement of the subspace number underlying $L_0$, i.e., CONO is dropping while $\mu_2(L_0)$ is increasing. This phenomenon additionally reflects that, besides of the low-rankness property, the extra structures (beyond low-rankness) also have a dramatic influence on the recovery of the latent matrix $L_0$.
\subsection{How to Choose the Dictionary in LRFD?}
As aforementioned, the first coherence parameter $\mu_1$ is invariant to the variations of the geometric distribution of data points. Hence, a promising direction for recovering non-uniform data might be to figure out in which conditions LRFD can \emph{avoid} the influences of the second coherence parameter $\mu_2$. We shall show that, when the dictionary $A$ itself is low-rank, LRFD is able to get around of $\mu_2$. Namely, the following two theorems are proved without using $\mu_2$ (The detailed procedures of proof can be found in Section~\ref{sec:proof}).

\begin{theorem}[Noiseless]\label{thm:noiseless}
Let $U_0\Sigma_0V_0^T$ be the SVD of $L_0$. Suppose that the dictionary matrix $A$ with SVD $U_A\Sigma_AV_A^T$ satisfies $\mathcal{P}_{U_A}(U_0)=U_0$ (i.e., $U_0$ is a subspace of $U_A$). For any $\delta>0$ and some numerical constant $c_a>0$, if
\begin{eqnarray*}
\rank{L_0}\leq\rank{A}\leq\frac{\delta^2n_2}{c_a\mu_1(A)\log{}n_1}\textrm{ and }\frac{|\Omega|}{mn}\geq\delta,
\end{eqnarray*}
then with probability at least $1-n_1^{-10}$, the optimal solution (denoted as $Z^*$) to problem \eqref{eq:mc} is unique and exact, in a sense that
$Z^*=A^+L_0$.
\end{theorem}

Figure~\ref{fig:dict} further confirms that there exist some kind of dictionaries using which LRFD is immune to the second coherence parameter $\mu_2$. The condition $\mathcal{P}_{U_A}(U_0)=U_0$ (i.e., $U_0$ is a subspace of $U_A$) is indispensable if we ask for the exactness of recovery, as $U_0\subset{}U_A$ is implied by the equality $AZ^*=L_0$. So what is suggested by above theorem is that the dictionary matrix $A$ should be made low-rank. This provides an elementary criterion for learning the dictionary matrix of LRFD.

\begin{figure}[h!]
\begin{center}
\includegraphics[width=0.95\textwidth]{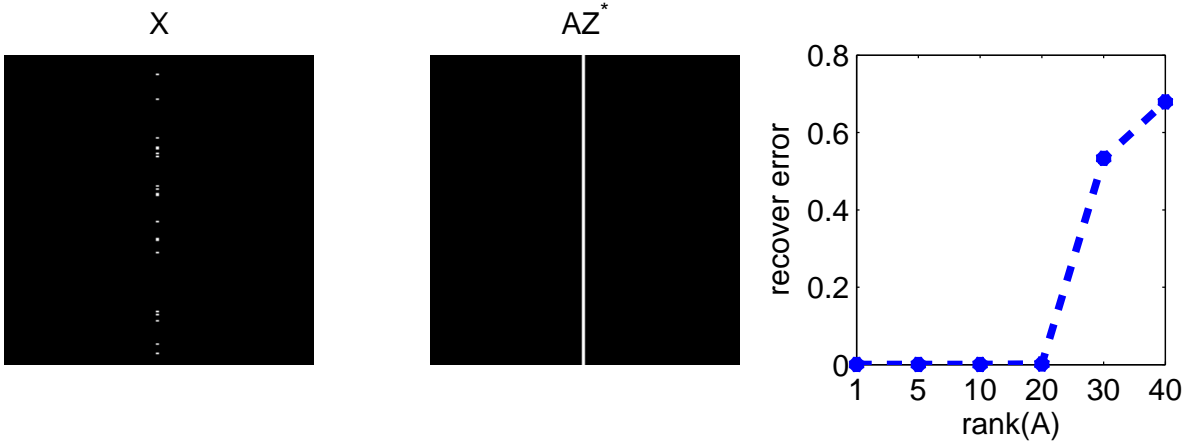}
\caption{Illustrating that LRFD can avoid $\mu_2$. In these experiments, $L_0$ is a $200\times200$ rank-1 matrix with one column be $\mathbf{1}$ (i.e, a vector of all ones) and everything else being zero. So $\mu_1(L_0)=1$ and $\mu_2(L_0)=n=200$. The dictionary in LRFD is set as $A=[\mathbf{1},W]$, where $W$ is a $200\times{}p$ random Gaussian matrix ($p$ is varying). The columns of $A$ are further normalized to have a unit length. As long as $\rank{A}\leq20$, the latent matrix $L_0$ (with high coherence) can be exactly recovered from an incomplete observation matrix $X$, 90\% entries of which are missing.}\label{fig:dict}
\end{center}
\end{figure}

The program \eqref{eq:cono} is designed for the case where the observed entries are noiseless. In reality this is often not true and the observations themselves could be actually contaminated. Cand{\`e}s and Plan have proven in~\cite{CandesPIEEE} that, even when the few observed entries are contaminated by a small amount of noise, matrix completion can be accurately performed by the following modified version:
\begin{eqnarray}\label{eq:cono:noisy}
\min_{L}\|L\|_*, &\textrm{s.t.}&\|\mathcal{P}_{\Omega}(X-L)\|_F\leq{}\epsilon,
\end{eqnarray}
where $\epsilon>0$ is a parameter that measures the noise level of the observations.

Similarly, LRFD \eqref{eq:mc} could be also modified to handle the problem of noisy matrix completion:
\begin{eqnarray}\label{eq:mc:noisy}
\min_{Z}\|Z\|_*, &\textrm{s.t.}&\|\mathcal{P}_{\Omega}(X-AZ)\|_F\leq{}\epsilon.
\end{eqnarray}
In the presence of dense noise, it is unrealistic to achieve exact recovery. Yet we have the following theorem to guarantee the recovery accuracy of \eqref{eq:mc:noisy}:
\begin{theorem}[Noisy]\label{thm:noisy}
Suppose that the dictionary matrix $A$ with SVD $U_A\Sigma_AV_A^T$ satisfies $\mathcal{P}_{U_A}(U_0)=U_0$ (i.e., $U_0\subset{}U_A$), and $\|\mathcal{P}_{\Omega}(X-L_0)\|_F\leq\epsilon$. For any $\delta>0$ and some numerical constant $c_a>0$, if
\begin{eqnarray*}
\rank{L_0}\leq\rank{A}\leq\frac{\delta^2n_2}{c_a\mu_1(A)\log{}n_1} \textrm{ and } \frac{|\Omega|}{mn}\geq2\delta,
\end{eqnarray*}
then with probability at least $1-n_1^{-10}$, the optimal solution (denoted as $Z^*$) to problem \eqref{eq:mc:noisy} gives a near recovery to $L_0$, in a sense that $\|AZ^*-L_0\|_F\leq2\epsilon/\delta$.
\end{theorem}
\subsection{An Algorithm for Matrix Completion}
The theorems introduced above provide a general direction for configuring the dictionary matrix in LRFD, implying several potential procedures. For example, one may drive some kind of optimization framework to jointly compute the variables $A$ and $Z$. In this paper, we would like to introduce a simple yet solid algorithm: We firstly obtain an estimate of $L_0$ by using CONO and then utilize the estimate to construct the dictionary matrix $A$ in LRFD. For the stability of computation, the CONO program \eqref{eq:cono:noisy} is implemented by solving its equivalent version:
\begin{eqnarray}\label{eq:cono:noisy1}
\min_{L}\|L\|_* + \frac{\lambda}{2}\|\mathcal{P}_{\Omega}(X-L)\|_F^2,
\end{eqnarray}
where $\lambda>0$ is taken as a parameter. Similarly, our LRFD program \eqref{eq:mc:noisy} is implemented by solving
\begin{eqnarray}\label{eq:mc:noisy1}
\min_{Z}\|Z\|_* + \frac{\lambda}{2}\|\mathcal{P}_{\Omega}(X-AZ)\|_F^2.
\end{eqnarray}
Provided that the observed entries are contaminated by small Gaussian noise and the dictionary $A$ is column-wisely unit-normed (i.e., $Ae_i=1,\forall{}i$), the regularization parameter $\lambda$ does not require extensive adjustments. Usually, $\lambda=100$ is a moderately good choice.
\begin{algorithm}[h!]
\caption{Matrix Completion}\label{alg:mc}
\begin{algorithmic}
\STATE \textbf{input:} An observed data matrix $X\in{}\mathbb{R}^{m\times{}n}$, and a support set $\Omega$ that stores the locations of the observed entries.
\STATE\textbf{adjustable parameter:} $\lambda$.
\STATE \textbf{1.} Solve for $\hat{L}_0$ by optimizing \eqref{eq:cono:noisy1} with $\lambda=100$.
\STATE \textbf{2.} Estimate the rank of $\hat{L}_0$ by $$\hat{r}_0=\#\{i: \sigma_i>10^{-3}\sigma_1\},$$
where $\sigma_1\geq\sigma_2\cdots$ are the singular values of $\hat{L}_0$.
\STATE \textbf{3.} Form $\tilde{L}_0$ by using the rank-$\hat{r}_0$ approximation of $\hat{L}_0$. That is,
\begin{eqnarray*}
\tilde{L}_0=\arg\min_{L}\|L-\hat{L}_0\|_F^2, \textrm{ s.t. } \rank{L}\leq\hat{r}_0,
\end{eqnarray*}
which is solved by SVD.
\STATE \textbf{4.} Construct a dictionary $\hat{A}$ from $\tilde{L}_0$ by normalizing the column vectors of $\tilde{L}_0$:
\begin{eqnarray*}
[\hat{A}]_{:,i}=\frac{[\tilde{L}_0]_{:,i}}{\|[\tilde{L}_0]_{:,i}\|_2}, i=1,\cdots,n,
\end{eqnarray*}
where $[\cdot]_{:,i}$ denotes the $i$th column of a matrix.
\STATE\textbf{5.} Solve for $Z^*$ by optimizing problem \eqref{eq:mc:noisy1} with $A=\hat{A}$ and $\lambda=100$.
\STATE \textbf{output:} $\hat{A}Z^*$.
\end{algorithmic}
\end{algorithm}

Algorithm \ref{alg:mc} summarizes the whole procedure of our algorithm for matrix completion. Note that the post-processing steps (Step 2 and Step 3) that mildly process the solution of CONO is to further encourage low-rank and well-conditioned dictionary, which is a sufficient condition for LRFD to succeed. To facilitate the choice of the parameter $\lambda$, Step 4 further normalizes the column vectors and ensure that the produced dictionary is column-wisely unit-normed.

While simple, our Algorithm \ref{alg:mc} is guaranteed in theory not to regress backward. That is, whenever CONN has already been successful in recovering $L_0$, the claims made in Theorem \ref{thm:noiseless} and Theorem \ref{thm:noisy} imply that the recovery produced by Algorithm \ref{alg:mc} is successful too.
\section{Experiments}
\subsection{Results on Randomly Generated Data}
We first verify the effectiveness of our Algorithm \ref{alg:mc} on randomly generated matrices. We generate a collection of $200\times1000$ data matrices according to the model of $X=\mathcal{P}_{\Omega}(L_0)$: $\Omega$ is an index set chosen at random, and $L_0$ is created by sampling 200 data points from each of 10 randomly generated subspaces. The rank of each subspace varies from 1 to 20 with step size 1, and thus the rank of $L_0$ varies from 10 to 200 with step size 10. The observation fraction $|\Omega|/(mn)$ varies from 32.5\% to 80\% with step size 2.5\%. For each combination of rank and observation fraction, we run 10 trials, resulting in a total number of 4000 ($20\times20\times10$) trials.

\begin{figure}[h!]
\begin{center}
\includegraphics[width=0.95\textwidth]{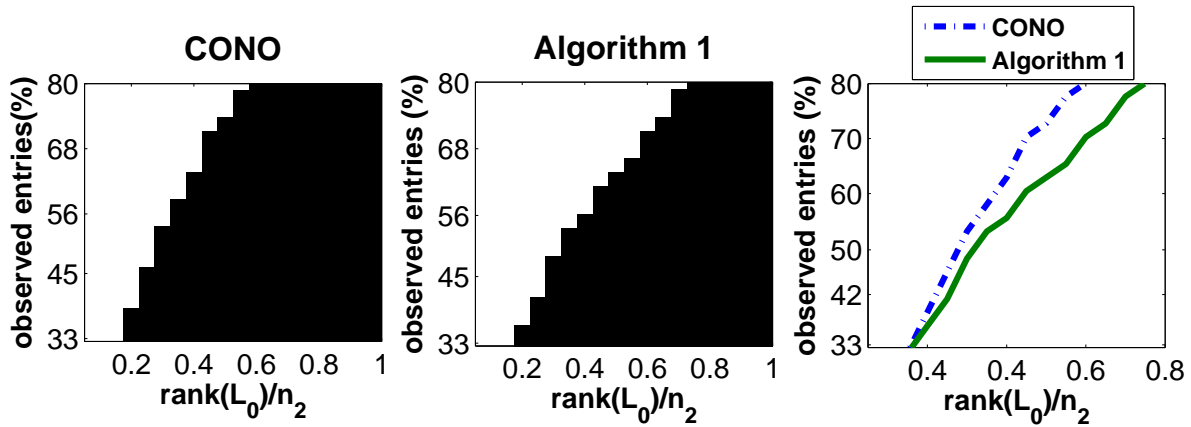}
\caption{\textbf{Comparing CONN with Algorithm \ref{alg:mc} on randomly generated matrices.} A curve shown in the third column is the boundary for an algorithm to be successful. In other words, the recovery is successful for any pair $(\rank{L_0}/n_2,|\Omega|/(mn))$ above the curve. Here, the success of recovery is in a sense that $\|\hat{L}_0-L_0\|_F<0.05\|L_0\|_F$, where $\hat{L}_0$ denotes an estimate of $L_0$. }\label{fig:recover}
\end{center}
\end{figure}

Figure~\ref{fig:recover} compares our Algorithm \ref{alg:mc} to CONN, both using $\lambda=10^6$. It can be seen that the learnt dictionary matrix works distinctly better than the identity matrix adopted by CONN. Namely, the area of the success region (i.e., white region) of our algorithm is 24.6\% larger than that of CONN. This verifies the significance of dictionary learning and the effectiveness of our Algorithm \ref{alg:mc}.
\subsection{Results on Motion Data}\label{sec:exp:motion}
We now experiment by using real motion sequences with incomplete trajectories. We use 11 additional sequences attached to the Hopkins155~\cite{hopkin155} database. Each sequence is a sole dataset (i.e., data matrix) and so there are in total 11 datasets of different properties, including the number of subspaces, the data dimension and the number of data samples. Particularly, in those sequences about 10\% of the entries in the data matrix of trajectories are unobserved (i.e., missed) due to vision occlusion, as illustrated in Figure~\ref{fig:motion}.

Notice that the ground truth matrix $L_0$ is unknown. To evaluate matrix completion algorithms in a quantitative way, we use the \emph{clustering error rates} produced by existing subspace clustering methods as the metrics to evaluate the quality of matrix completion. Namely, we firstly perform subspace clustering on both the incomplete trajectory matrices and the completed versions, and then compute the clustering error rates of the existing subspace clustering methods. We consider three state-of-the-art subspace clustering methods, including Shape Interaction Matrix (SIM)~\cite{ijcv_1998_factor}, Low-Rank Representation with $dictionary=X$ (LRRx)~\cite{icml_2010_lrr} and Sparse Subspace Clustering (SSC)~\cite{cvpr_2009_ssc}. As none of these methods owns a mechanism for handling the missing entries, we implement a simple strategy for them: Each missed entry is nominally assigned a value of zero.
\begin{figure}[h!]
\begin{center}
\includegraphics[width=0.8\textwidth]{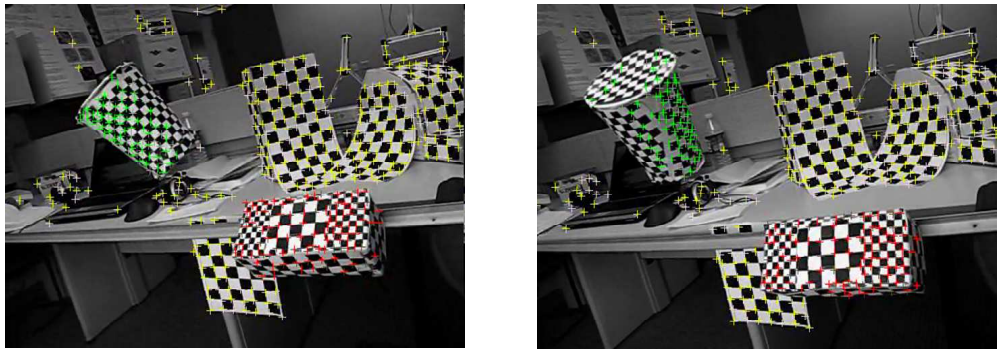}
\caption{Example image frames from the motion sequences used in our experiments. Due to the rotation of the objects, some measurements in the data matrix of trajectories are missing.}\label{fig:motion}
\end{center}
\end{figure}

Table \ref{tb:motion} shows the error rates of various algorithms. Without the preprocessing of matrix completion, all the subspace clustering methods fail to accurately categorize the trajectories of motion objects, producing error rates higher than 19\%. In contrast, without the presence of missing entries, the lowest error rate from SIM, LRRx and SSC on Hopkins155 is as low as 1\%~\cite{tpami_2013_lrr}. This illustrates that it is important for motion segmentation to restore the missing entries possibly existing in the data matrix of trajectories. By using CONO (with $\lambda=100$) to restore the missing entries, the clustering performances of all considered subspace clustering methods are improved dramatically. For example, the error rate of SSC is reduced from 31.75\% to 3.24\%. By seeking an advanced solution for matrix completion using Algorithm~\ref{alg:mc} (with $\lambda=100$), the error rates can be reduced again. For example, the error rate of LRRx is reduced from $7\%$ to $5\%$, which is a 28\% improvement. These results verify the effectiveness of our dictionary learning strategy in a realistic environment.
\begin{table}[h!]
\caption{Clustering error rates (\%) on the 11 motion sequences with incomplete trajectories.}\label{tb:motion}
\begin{center}
\begin{tabular}{|l|ccc|}\hline
        &SIM            &CONN+ SIM                &Algorithm \ref{alg:mc}+ SIM\\\hline
mean	&19.70		    &12.04				      &10.76		\\
max		&40.04		    &44.76				      &38.97		\\
min		&3.27			&0.58					  &0.45			\\
std		&11.47		    &15.38				      &11.79		\\
time (sec.)&0.05	    &8.42			          &12.34	    \\\hline\hline
        &LRRx           &CONN+LRRx                &Algorithm \ref{alg:mc}+LRRx  \\\hline
mean	&19.85		    &7.06					  &4.94			\\
max		&36.83		    &49.68				      &22.22		\\
min		&0.90			&0.33					  &0.33			\\
std		&14.66		    &14.38				      &6.54			\\
time (sec.)&2.92		&9.42				      &13.03		\\\hline\hline
        &SSC            &CONN+SSC                 &Algorithm \ref{alg:mc}+SSC  \\\hline
mean	&31.75		    &3.24					  &2.98			\\
max		&47.19		    &12.21				      &10.28		\\
min		&19.25		    &0					      &0			\\
std		&10.24		    &4.24					  &3.87			\\
time (sec.)&2.33	    &10.47				      &14.32		\\\hline
\end{tabular}
\end{center}
\end{table}
\section{Conclusion and Future Work} This paper pointed out that there could exist rich structures inside a low-dimensional subspace, so called as \emph{extra structures beyond low-rankness}. We showed that such extra structures cannot be ignored and have dramatic influences on the success of restoring a low-rank matrix form the incomplete versions. We further proposed a novel model termed LRFD (Low-Rank Factor Decomposition) which handle the extra structures by imposing an additional constraint that the data points are represented by the linear combinations of the bases of a dictionary. Provided that the dictionary is configured properly, LRFD could generally work well on non-uniform data without knowing an precise model of the geometric distributions of the data points. We mathematically proved some theorems which suggest that the dictionary matrix in LRFD should be made low-rank. Subsequently, we established a brief algorithm for approximating such dictionaries in unsupervised environments. Extensive simulations and experiments verify the effectiveness of our algorithm.

The goal of this paper is to analyze in general the problem of modeling extra structures beyond low-rankness and provide some basic principles for resolving the problem. Our proposed algorithm does not aim at completely solving the problem, but rather target on a small yet solid step for advancing matrix completion. It is entirely possible to develop more effective algorithms for learning the dictionary matrix in LRFD and we leave this as future work.

\section*{Acknowledgement}

Guangcan Liu is a Postdoctoral Researcher supported by NSF-DMS0808864,  NSF-EAGER1249316, AFOSR-FA9550-13-1-0137, and ONR-N00014-13-1-0764. Ping Li is also partially supported by  NSF-III1360971 and NSF-BIGDATA1419210.

\newpage\clearpage

\section{Mathematical Proofs}\label{sec:proof}
\subsection{Proof of Theorem \ref{thm:noiseless}}
The same as in CONO, we assume that the locations of the observed entries are selected uniformly at random. In more details, we work with the Bernoulli model $\Omega=\{(i,j):\delta_{ij}=1\}$, where $\delta_{ij}$'s are i.i.d variables taking value one with probability $\rho_0=|\Omega|/(mn)$ and zero with probability $(1-\rho_0)$, so that the expected cardinality of $\Omega$ is $\rho_0mn$.

We first establish the following lemma that plays an important role in the proof.
\begin{lemma}\label{lem:pupo}
Suppose $\Omega\sim{}Ber(\rho_0)$. Then for any $\delta>0$, we have
\begin{eqnarray*}
\|\mathcal{P}_{U_A}\mathcal{P}_{\Omega^\bot}\mathcal{P}_{U_A}\|\leq{}1-\rho_0+\delta
\end{eqnarray*}
obeys with probability at least $1-n_1^{-10}$, provided that
\begin{eqnarray*}
\rank{A}\leq\frac{\delta^2n_2}{c_a\mu_1(A)\log{}n_1},
\end{eqnarray*}
where $c_a$ is a numerical constant.
\end{lemma}
\begin{proof}
For any matrix $M$, we have
\begin{eqnarray*}
\mathcal{P}_{U_A}(M)=\sum_{i,j}\langle{}\mathcal{P}_{U_A}(M),e_ie_j^T\rangle{}e_ie_j^T,
\end{eqnarray*}
and so
\begin{eqnarray*}
\mathcal{P}_{\Omega}\mathcal{P}_{U_A}(M)=\sum_{i,j}\delta_{ij}\langle{}\mathcal{P}_{U_A}(M),e_ie_j^T\rangle{}e_ie_j^T,
\end{eqnarray*}
which gives
\begin{eqnarray*}
\mathcal{P}_{U_A}\mathcal{P}_{\Omega}\mathcal{P}_{U_A}(M)&=&\sum_{i,j}\delta_{ij}\langle{}\mathcal{P}_{U_A}(M),e_ie_j^T\rangle{}\mathcal{P}_{U_A}(e_ie_j^T)\\
&=&\sum_{i,j}\delta_{ij}\langle{}M,\mathcal{P}_{U_A}(e_ie_j^T)\rangle{}\mathcal{P}_{U_A}(e_ie_j^T).
\end{eqnarray*}
Note that the Frobenius norm of a matrix is equivalent to the vector $\ell_2$ norm, while considering the matrix as a long vector. In that sense, we have
\begin{eqnarray*}
\mathcal{P}_{U_A}\mathcal{P}_{\Omega}\mathcal{P}_{U_A}=\sum_{i,j}\delta_{ij}\mathcal{P}_{U_A}(e_ie_j^T)\otimes\mathcal{P}_{U_A}(e_ie_j^T),
\end{eqnarray*}
where $\otimes$ denotes the Kronecker product.

The definition of $\mu_1(A)$ gives
\begin{eqnarray*}
\|\mathcal{P}_{U_A}(e_ie_j^T)\|_F^2\leq{}\frac{\mu_1(A)r_A}{m},
\end{eqnarray*}
where $r_A\equiv\rank{A}$. Then by using the results in \cite{Rudelson99randomvectors} and following the proof procedure in \cite{Candes:2009:math}, it could be concluded that the inequality
\begin{eqnarray*}
&&\|\rho_0\mathcal{P}_{U_A}-\mathcal{P}_{U_A}\mathcal{P}_{\Omega}\mathcal{P}_{U_A}\|\leq\rho_0(\phi_1\sqrt{\frac{\mu_1(A)r_A\log{}n_1}{n_2}}\\
&&+\phi_2\sqrt{\frac{\mu_1(A)\beta{}r_A\log{}n_1}{n_2}})\leq\phi_1\sqrt{\frac{\mu_1(A)r_A\log{}n_1}{n_2}}\\
&&+\phi_2\sqrt{\frac{\mu_1(A)\beta{}r_A\log{}n_1}{n_2}}
\end{eqnarray*}
obeys with probability at least $1-n_1^{-\beta}$ for some numerical constants $\phi_1$ and $\phi_2$. For any $\delta>0$, setting $\beta=10$ and $c_a=(\phi_1+\sqrt{10}\phi_2)^2$ gives that
\begin{eqnarray*}
\|\rho_0\mathcal{P}_{U_A}-\mathcal{P}_{U_A}\mathcal{P}_{\Omega}\mathcal{P}_{U_A}\|\leq\delta
\end{eqnarray*}
holds with probability at least $1-n_1^{-10}$, provided that $r_A\leq{}\delta^2n_2/(c_a\mu_1(L_0)\log{}n_1)$.

By the equality that $\mathcal{P}_{U_A}\mathcal{P}_{\Omega^\bot}\mathcal{P}_{U_A}=(1-\rho_0)\mathcal{P}_{U_A}+(\rho_0\mathcal{P}_{U_A}-\mathcal{P}_{U_A}\mathcal{P}_{\Omega}\mathcal{P}_{U_A})$ and the triangle inequality,
\begin{eqnarray*}
&&\|\mathcal{P}_{U_A}\mathcal{P}_{\Omega^\bot}\mathcal{P}_{U_A}\|\leq\|(1-\rho_0)\mathcal{P}_{U_A}\|\\
&&+\|\rho_0\mathcal{P}_{U_A}-\mathcal{P}_{U_A}\mathcal{P}_{\Omega}\mathcal{P}_{U_A}\|\leq1-\rho_0+\delta.
\end{eqnarray*}
\end{proof}

Based on the above lemma, we easily prove the following lemma which states that $(P_{U_A}P_{\Omega}P_{U_A})^{-1}$ is well defined and has a small operator norm.
\begin{lemma}\label{lem:inverse}
Let $\|P_{U_A}P_{\Omega^\bot}P_{U_A}\|=\psi$. If $\psi<1$, then the operator $\mathcal{P}_{U_A}\mathcal{P}_{\Omega}\mathcal{P}_{U_A}$ is an injection from $\mathcal{P}_{U_A}$ to $\mathcal{P}_{U_A}$, and its inverse operator is given by
\begin{eqnarray*}
\mathcal{I}+\sum_{i=1}^{\infty}(\mathcal{P}_{U_A}\mathcal{P}_{\Omega^\bot}\mathcal{P}_{U_A})^i.
\end{eqnarray*}
\end{lemma}
\begin{proof}
By $\|\mathcal{P}_{U_A}\mathcal{P}_{\Omega^\bot}\mathcal{P}_{U_A}\|=\psi<1$, we have that $\mathcal{I}+\sum_{i=1}^{\infty}(\mathcal{P}_{U_A}\mathcal{P}_{\Omega^\bot}\mathcal{P}_{U_A})^i$ is well defined and has an operator norm not larger than $1/(1-\psi)$.

Note that
\begin{eqnarray*}
&&\mathcal{P}_{U_A}\mathcal{P}_{\Omega}\mathcal{P}_{U_A}=\mathcal{P}_{U_A}(\mathcal{I}-\mathcal{P}_{\Omega^\bot})\mathcal{P}_{U_A}\\
&&=\mathcal{P}_{U_A}(\mathcal{I}-\mathcal{P}_{U_A}\mathcal{P}_{\Omega^\bot}\mathcal{P}_{U_A}).
\end{eqnarray*}
Thus for any $M\in\mathcal{P}_{U_A}$ the following holds:
\begin{eqnarray*}
&&\mathcal{P}_{U_A}\mathcal{P}_{\Omega}\mathcal{P}_{U_A}(\mathcal{I}+\sum_{i=1}^{\infty}(\mathcal{P}_{U_A}\mathcal{P}_{\Omega^\bot}\mathcal{P}_{U_A})^i)(M)\\
&=&\mathcal{P}_{U_A}(\mathcal{I}-\mathcal{P}_{U_A}\mathcal{P}_{\Omega^\bot}\mathcal{P}_{U_A})\\
&&(\mathcal{I}+\sum_{i=1}^{\infty}(\mathcal{P}_{U_A}\mathcal{P}_{\Omega^\bot}\mathcal{P}_{U_A})^i)(M)\\
&=&\mathcal{P}_{U_A}(\mathcal{I}+\sum_{i=1}^{\infty}(\mathcal{P}_{U_A}\mathcal{P}_{\Omega^\bot}\mathcal{P}_{U_A})^i-\mathcal{P}_{U_A}\mathcal{P}_{\Omega^\bot}\mathcal{P}_{U_A}\\
&&-\sum_{i=2}^{\infty}(\mathcal{P}_{U_A}\mathcal{P}_{\Omega^\bot}\mathcal{P}_{U_A})^i)(M)\\
&=&\mathcal{P}_{U_A}(\mathcal{I}+\sum_{i=1}^{\infty}(\mathcal{P}_{U_A}\mathcal{P}_{\Omega^\bot}\mathcal{P}_{U_A})^i\\
&&-\sum_{i=1}^{\infty}(\mathcal{P}_{U_A}\mathcal{P}_{\Omega^\bot}\mathcal{P}_{U_A})^i)(M)\\
&=&\mathcal{P}_{U_A}(M) = M.
\end{eqnarray*}
\end{proof}

The next lemma finishes to prove Theorem \ref{thm:noiseless}.
\begin{lemma}\label{lem:dual}
If $\|P_{U_A}P_{\Omega^\bot}P_{U_A}\|<1$, which follows from $|\Omega|>\delta{}mn$, then $Z^*=A^+L_0$ is the unique optimal solution to the convex optimization problem \eqref{eq:mc}.
\end{lemma}
\begin{proof}
By $U_0\subset{}U_A$, $Z^*=A^{+}L_0$ is feasible to \eqref{eq:mc}. By standard convexity arguments \cite{book:convex}, $Z^*=A^{+}L_0$ is an optimal solution to \eqref{eq:mc} if there exists a dual vector (or Lagrange multiplier) $Y$ such that
\begin{eqnarray*}
A^T\mathcal{P}_{\Omega}(Y)\in\partial\|A^{+}L_0\|_*,
\end{eqnarray*}
where $\partial(\cdot)$ is the sub-gradient of a function. Let the SVD of $A^{+}L_0$ be $U\Sigma{}V^T$. Then we define $Y$ as
\begin{eqnarray*}
Y = \mathcal{P}_{\Omega}\mathcal{P}_{U_A}(\mathcal{I}+\sum_{i=1}^{\infty}(\mathcal{P}_{U_A}\mathcal{P}_{\Omega^\bot}\mathcal{P}_{U_A})^i)((A^T)^+UV^T).
\end{eqnarray*}
With this notation, we have
\begin{eqnarray*}
A^T\mathcal{P}_{\Omega}(Y)&=&A^T\mathcal{P}_{U_A}\mathcal{P}_{\Omega}(Y)\\
&=&A^T\mathcal{P}_{U_A}\mathcal{P}_{\Omega}\mathcal{P}_{U_A}\\
&&(\mathcal{I}+\sum_{i=1}^{\infty}(\mathcal{P}_{U_A}\mathcal{P}_{\Omega^\bot}\mathcal{P}_{U_A})^i)((A^T)^+UV^T)\\
&=&A^T(A^T)^+UV^T=V_AV_A^TUV^T\\
&=&UV^T\in\partial\|A^{+}L_0\|_*,
\end{eqnarray*}
which gives that $Z=A^+L_0$ is an optimal solution to the convex optimization problem \eqref{eq:mc}.

It remains to prove that the optimal solution to \eqref{eq:mc} is unique. We shall consider a feasible perturbation $Z=A^{+}L_0+\Delta$ and show that the objective strictly increases whenever $\Delta\neq0$. By
\begin{eqnarray*}
0&=&\mathcal{P}_{\Omega}(X-A(A^{+}L_0))\\
&=&\mathcal{P}_{\Omega}(X-A(A^{+}L_0+\Delta)),
\end{eqnarray*}
we have
\begin{eqnarray*}
\mathcal{P}_{\Omega}(A\Delta) = 0, &\textrm{i.e.,}& A\Delta\in\Omega^c.
\end{eqnarray*}
Then, by $A\Delta\in{}U_A$, we have $A\Delta\in{}\Omega^c\cap{}U_A$. This, together with the assumption $\Omega^c\cap{}U_A=\{0\}$, gives
\begin{eqnarray*}
A\Delta = 0,&\textrm{i.e.,}& \Delta\in{}V_A^\bot\subset{}U^\bot,
\end{eqnarray*}
where $(\cdot)^\bot$ denotes the orthogonal complement of an orthonormal matrix.

We also have
\begin{eqnarray*}
\|A^+L_0+\Delta\|_*&=&\|\left[\begin{array}{c}
U^T\\
(U^{\bot})^T\\
\end{array}
\right]
(A^+L_0+\Delta)[V,V^{\bot}]\|_*\\
&=&\norm{\left[\begin{array}{cc}
U^TA^+L_0V&0\\
(U^{\bot})^T\Delta{}V&(U^{\bot})^T\Delta{}V^{\bot}\\
\end{array}\right]}_*\\
&\geq&\norm{U^TA^+L_0V}_*\\
&=&\|A^+L_0\|_*,
\end{eqnarray*}
where the equality can hold if and only if
\begin{eqnarray*}
(U^{\bot})^T\Delta{}V=0 &\textrm{and}&(U^{\bot})^T\Delta{}V^{\bot}.
\end{eqnarray*}
This gives $(U^{\bot})^T\Delta=0$, i.e., $\Delta\in{}U$. However, we have already proven $\Delta\in{}U^{\bot}$. Thus, the inequality $\|A^+L_0+\Delta\|_*>\|A^+L_0\|_*$ strictly holds unless $\Delta=0$. In other words, $Z^*=A^+L_0$ is the unique optimal solution to \eqref{eq:mc}.
\end{proof}
\subsection{Proof of Theorem \ref{thm:noisy}}
\begin{proof}
By triangle inequality,
\begin{eqnarray*}
\|\mathcal{P}_{\Omega}(AZ^*-L_0)\|_F&=&\|\mathcal{P}_{\Omega}(AZ^*-X)\\
&&+\mathcal{P}_{\Omega}(X-L_0)\|_F\\
&\leq&\|\mathcal{P}_{\Omega}(AZ^*-X)\|_F\\
&&+\|\mathcal{P}_{\Omega}(X-L_0)\|_F\\
&\leq&2\epsilon.
\end{eqnarray*}
Since $U_0\subset{}U_A$, $AZ^*-L_0\in\mathcal{P}_{U_A}$. By the invertibility of $\mathcal{P}_{U_A}\mathcal{P}_{\Omega}\mathcal{P}_{U_A}$,
\begin{eqnarray*}
AZ^*-L_0&=&(\mathcal{I}+\sum_{i=1}^{\infty}(\mathcal{P}_{U_A}\mathcal{P}_{\Omega^\bot}\mathcal{P}_{U_A})^i)\\
&&\mathcal{P}_{U_A}\mathcal{P}_{\Omega}\mathcal{P}_{U_A}(AZ^*-L_0),
\end{eqnarray*}
where the validity (with probability at least $1-n_1^{-10}$) of $\mathcal{I}+\sum_{i=1}^{\infty}(\mathcal{P}_{U_A}\mathcal{P}_{\Omega^\bot}\mathcal{P}_{U_A})^i$ is from Lemma \ref{lem:inverse}.

It could be calculated that
\begin{eqnarray*}
\|AZ^*-L_0\|_F&=&\|(\mathcal{I}+\sum_{i=1}^{\infty}(\mathcal{P}_{U_A}\mathcal{P}_{\Omega^\bot}\mathcal{P}_{U_A})^i)\\
&&\mathcal{P}_{U_A}\mathcal{P}_{\Omega}\mathcal{P}_{U_A}(AZ^*-L_0)\|_F\\
&\leq&\|(\mathcal{I}+\sum_{i=1}^{\infty}(\mathcal{P}_{U_A}\mathcal{P}_{\Omega^\bot}\mathcal{P}_{U_A})^i)\|\\
&&\times\|\mathcal{P}_{U_A}\mathcal{P}_{\Omega}\mathcal{P}_{U_A}(AZ^*-L_0)\|_F\\
&\leq&\frac{\|\mathcal{P}_{U_A}\mathcal{P}_{\Omega}\mathcal{P}_{U_A}(AZ^*-L_0)\|_F}{\rho_0-\delta}\\
&=&\frac{\|\mathcal{P}_{U_A}\mathcal{P}_{\Omega}(AZ^*-L_0)\|_F}{\rho_0-\delta}\\
&\leq&\frac{\|\mathcal{P}_{\Omega}(AZ^*-L_0)\|_F}{\rho_0-\delta}\\
&\leq&\frac{2\epsilon}{\rho_0-\delta}\\
&\leq&\frac{2\epsilon}{\delta},
\end{eqnarray*}
where the last inequality is concluded from the condition $\rho_0=|\Omega|/(mn)\geq{}2\delta$.
\end{proof}

\end{document}